\documentclass[11pt]{article}
\usepackage[utf8x]{inputenc}

\usepackage[a4paper,top=2cm, bottom=2cm, left=2.2cm, right=2.2cm]{geometry}

\usepackage{amsmath,amsthm}
\usepackage{amssymb}
\usepackage{amsfonts}
\usepackage{enumerate}
\usepackage{rotating}
\usepackage{caption}
\usepackage{setspace}
\usepackage{footmisc}
\usepackage{graphicx,booktabs}
\usepackage{subcaption}
\usepackage{float}
\usepackage{bm}
\usepackage{color}
\usepackage{tabularx,tabu}
\usepackage{multirow}
\usepackage{placeins}
\usepackage{dcolumn}
\newcolumntype{d}[1]{D{.}{.}{#1}}
\usepackage{hyperref}
\usepackage{setspace}
\usepackage[colorinlistoftodos]{todonotes}
\usepackage{dirtytalk}
\usepackage{natbib}

\allowdisplaybreaks[1]
\makeatletter
\renewcommand*{\@fnsymbol}[1]{\@alph{\numexpr#1}}
\makeatother

\doublespace

\def\citetpos#1{\citeauthor{#1}'s (\citeyear{#1})}
\def\citeauthorpos#1{\citeauthor{#1}'s}
\defcitealias{FRWZ2015}{FRWZ}  
\hypersetup{colorlinks=true,citecolor=blue,linkcolor=red,filecolor=green,urlcolor=red}

\newtheorem{proposition}{Proposition}

\theoremstyle{remark}

\begin{document}
\title{Monetary-fiscal interactions under price level targeting}
\author{Guido Ascari\thanks{Corresponding author: Department of Economics, University of Oxford, Manor Road, Oxford OX1 3UQ, United Kingdom. E-mail address: guido.ascari@economics.ox.ac.uk}\\
\textit{University of Oxford} 
\\[-3pt]\textit{University of Pavia}
\and 
Anna Florio\thanks{Department of Management,
Economics and Industrial Engineering, Politecnico di Milano, via Lambruschini
4/B, 20156 Milan, Italy.
E-mail address: anna.florio@polimi.it}\\
\textit{Politecnico di Milano}
\and Alessandro Gobbi\thanks{Department of Environmental Science and Policy, University of Milan, Via Celoria, 2,  20133 Milan, Italy.
E-mail address: alessandro.gobbi@unimi.it}\\
\textit{University of Milan}
}
\date{September 2020}
\maketitle

\begin{abstract}
The adoption of a ``makeup'' strategy is one of the proposals in the ongoing review of the Fed's monetary policy framework. Another suggestion, to avoid the zero lower bound, is a more active role for fiscal policy. We put together these ideas to study monetary-fiscal interactions under price level targeting. Under price level targeting and a fiscally-led regime, we find that following a deflationary demand shock: (i) the central bank increases (rather than decreases) the policy rate; (ii) the central bank, thus, avoids the zero lower bound; (iii) price level targeting is generally welfare improving if compared to inflation targeting.


\vspace{20pt} 


\end{abstract}
\thispagestyle{empty} 
\setcounter{page}{0}
\pagebreak



\section{Introduction}

After more than a decade from the outburst of the Great Recession that made the policy rate hit the lower limit of zero, the outbreak of the Covid-19 pandemic has revitalized the never faded debate on how to avoid the zero lower bound (ZLB). Central bankers are worried that the low level of the equilibrium real interest rates could make the ZLB a binding constraint, making it harder in the future to react to a deflationary shock.\footnote{Note that the neutral rate might further decline due to the long-lasting effects of the pandemic, see \citet{Jorda2020}.}

 

One of the most discussed suggestions, especially in the U.S., has been the adoption of a (temporary) price level targeting (PLT) \citep{Bernanke2017,Bernanke2019}.\footnote{See on this point even \citet{Williams2017}, \citet{Evans2012} and \citet{Yellen2018}. 
Note that another often cited option to avoid the ZLB is increasing the inflation target. However, \citet{Andreade2019} find that the case for increasing the inflation target is much reduced when the central bank adopts a price level targeting approach.} According to this proposal, whenever the economy is in the proximity of the ZLB, the central bank should commit to a lower-for-longer rate to fully offset any shortfalls of inflation from target. 
At the Jackson Hole 2020 symposium, Jay Powell announced a change in the Fed's monetary procedure towards an average inflation target, which is a kind of PLT with a finite time window for measuring inflation.\footnote{PLT was adopted in the 1930s by the Sweden's Riksbank that, in the more recent past, has re-pondered that option, as the Bank of Canada did.} Therefore, there will be periods when inflation will overshoot the stated inflation target to make up for its previous undershooting.

In addition to the adoption of new monetary procedures, even before the outbreak of the coronavirus, many observers suggested to use fiscal measures to stimulate aggregate demand in presence of adverse shocks. \citet{EvansGuseHonka2008} claimed that aggressive monetary policy could not be sufficient to avoid deflationary spirals whereas it should be combined with aggressive fiscal policy. \citet{Draghi2019} invoked a more active role for fiscal policy since he imputed the slide into disinflation to an unbalanced macroeconomic policy mix (in favor of monetary policy). \citet{Lagarde2019} called for euro area fiscal support claiming that monetary policy could not be the only game in town. To mitigate the negative effects of the pandemic, in most countries governments actually adopted unprecedented fiscal stimulus.

Especially in difficult times, then, monetary and fiscal policy interactions becomes crucial. \citet{Bernanke2017}, talking about PLT, states that central bank independence should be protected without \say{ruling out temporary periods of monetary-fiscal coordination that may be essential for achieving key policy goals.} When interest rates (are about to) hit zero and the price level is falling, monetary policy can completely lose control on inflation and, according to \citet{Sims2000}, an active fiscal policy could be beneficial.\footnote{According to \citet{Leeper:1991}'s taxonomy, a policy authority is active when it pursues its objective unconstrained by the actions of the other authority; otherwise it is passive. Two regimes return determinacy: the monetary-led regime (active monetary/passive fiscal) and the fiscally-led regime (passive monetary/active fiscal).} According to \citet{Sims2000}, two dynamics are at work during a deflation episode: an accelerationist dynamic that makes prices decrease further, and real balance effects that instead increase prices. \say{For real balance effects to rule out or prevent deflationary spirals, fiscal policy must be seen not to be committed to keeping the real value of primary surpluses in line with the current outstanding real value of government debt so that this will not be backed by increased future real taxation. Policy-makers should understand that under some circumstances budget balancing can become bad policy.} 

The implications of adopting a PLT for different monetary and fiscal policy regimes are, however, not studied so far in the literature. 
On the one hand, several studies have discussed the potential advantages of price-level stabilization over inflation stabilization in terms of improved welfare, lower inflation variability and a more favorable inflation-output gap trade-off \citep[see, among the others,][]{Giannoni2014,Svensson1999,Vestin2006}. Other studies theoretically justify a PLT approach to address the zero lower bound problem and to escape deflationary traps \citep[see][]{Evans2012, EggertsonWoodford2003,Billi2008}. All these studies, however, either abstract from fiscal policy or consider a passive fiscal policy (or a monetary-led regime).
On the other hand, a large literature analyses the effects of different monetary-fiscal mixes (both monetary-and fiscally-led regimes) to counteract recessions and deflationary traps in an inflation targeting (IT) context \citep[see, among the others,][]{DavigLeeper2011,BianchiMelosi2017,AFG2020}. There is no such an analysis, however, under PLT. We aim to fill this gap.

To our knowledge, this paper is the first attempt to merge the literature on monetary-fiscal interactions with the PLT proposal, while previous papers on PLT consider only a monetary-led regime framework. We investigate the dynamics of the model under PLT and a fiscally-led regime (i.e., active fiscal policy) and how the adoption of a PLT rule modifies Leeper's analysis of equilibrium uniqueness under rational expectations. In this respect, our paper could be considered as the counterpart under PLT of the analysis in \citet{Bhatta2014} under IT. 

Our main result is that under price level targeting and a fiscally-led regime, the central bank increases (rather than decreases) the policy rate after a negative demand shock. This leads the economy out of the deflationary trap and allows the central bank to avoid the zero lower bound. This result resonates with the neo-Fisherian recipe to increase the interest rates in order to increase inflation, as argued by e.g., \citet{Uribe2018}. However, here the logic is very different from the neo-Fisherian perspective. First, the nominal interest rate increases under PLT and active fiscal policy because determinacy implies that the price level coefficient in the interest rate rule should be lower than zero. Hence, the nominal interest rate increases if the price gap is negative. Second, after a demand shock, this ``inverse'' reaction of the policy rate to inflation amplifies the wealth effects due to the fiscal theory of the price level. Despite there is nothing normative in our PLT rule, we show that PLT is generally welfare improving if compared to IT after a demand shock, under both a monetary and a fiscally-led regime.

The paper is organized as follows. Section 2 introduces two of the main proposals---namely price level targeting and active fiscal policies---put forth to adjust the actual policy framework to face the current and future recessions. Section 3 shows how \citetpos{Leeper:1991} results change when the central bank adopts a PLT rule. Section 4 extends the determinacy analysis to a simple DSGE model and shows the effects of structural shocks both in the monetary- and in the fiscally-led regimes, comparing the impulse response functions under a PLT rule to those under a IT rule. This section includes a welfare analysis of the two approaches in case of a demand shock and it is closed with a description of the effects of that shock when the ZLB is a potentially binding constraint. Section 5 concludes. 

\section{Price level targeting and active fiscal policies}

Before the coronavirus crisis, the Federal Reserve was discussing how to review its inflation-targeting framework. The reasons for this choice were low inflation, that had run below the 2\% inflation target for much of the 2009-2020 expansion, falling inflation expectations and low interest rates that increased the risk of falling back into a ZLB episode, which readily occurred when the U.S. was hit by the Covid-19 shock. This last shock has speeded up the Fed's framework review leading the U.S. central bank to adopt, since August 2020, an average inflation target.
A central bank under IT aims to reach 2\% inflation every period and commits to correct the inflation deviations from the intended target. Under a makeup strategy, like PLT or average inflation targeting, instead, the central bank should reverse previous shocks to the price level, abandoning the \say{bygones are bygones} approach. In particular,  under PLT, the monetary authority defines a target path for the price level ($p_t^*$, in logs) and tries to correct deviations of the price level ($p_t$) from this path. The corresponding log-linear monetary policy rule has the form: 
\begin{equation}
r_t=\phi_p(p_t-p_t^{\ast })+\theta_t,\label{eq: PLTrule}
\end{equation}
where $r_t$ is the nominal interest rate (in deviations from the steady state) and $\theta_t$ is a monetary policy shock. Using the fact that $p_t^*=p_{t-1}^*+\pi^*$, where $\pi^*$ is the inflation target, one gets the following Wicksellian rule:
\begin{equation} \label{eq:wicks_rule}
r_t=\phi_p\left(\pi_t-\pi^*\right)+r_{t-1}+\Delta\theta_t.
\end{equation}
We will employ this rule throughout the paper when referring to PLT. 

The main traditional argument in favour of PLT is the long-run predictability of the price level, while the major disadvantage is the increase in the short-run variability of both inflation and output \citep[e.g.,][]{Lebow1992,HaldaneSalmon1995}. 
If the PLT is credible, the history-dependence introduced by this approach could achieve better economic outcomes thanks to the management of market expectations. After a deflationary shock, the central bank adopting a PLT approach commits to keep rates lower for longer to tackle deflationary expectations. This increases inflation expectations that, in turn, reduce real rates and scale down the risk of hitting the ZLB.\footnote{Throughout the paper, we assume central bank credibility and we do not cope with the time inconsistency problems that could well undermine it.}


The current situation, with low inflation \citep[e.g.,][]{Blanchard2020Vox}, huge fiscal stimulus and central banks that accommodate increased fiscal spending, could be akin to a period of fiscal dominance. 
What would happen if, in the presence of an active fiscal policy, the proposal of a PLT approach would be put forth? At the moment there is not a framework that links PLT and fiscally-led regimes.\footnote{The only exception is the so-called “going direct” approach \citep[see][] {Bartsch2019} that advocates a more explicit coordination between monetary and fiscal policy to be undertaken just in some defined circumstances, with an explicit exit strategy and with an explicit inflation objective for which both monetary and fiscal authorities are held accountable. The inflation target should be met through a monetary financing of a fiscal expansion (a fiscally-led regime) and the central bank should make up for past inflation shortfalls (a make-up strategy in line with PLT).} In the next section we introduce PLT in \citet{Leeper:1991}'s pivotal model of
monetary-fiscal interactions.

\section{Leeper's (1991) with PLT}

\citet{Leeper:1991} employs a flexible-price model to analyse the determinacy and the properties of equilibria produced by monetary and fiscal policy rules according to whether policy authorities behave actively or passively. We extend his model to consider a monetary rule that responds to the price level, rather than to inflation. The main reason to consider this simple \citetpos{Leeper:1991} setup is that PLT introduces an extra endogenous state variable so getting analytical results becomes impossible for a sticky price model. Already in this simple setup, under PLT and active fiscal policy, there are two endogenous state variable to consider (public debt and the nominal interest rate). Despite that, we were able to get the analytical solution of the model to gain some insights that carry over to the sticky price model.

Beside the monetary policy rule, the original model by \citet{Leeper:1991} is composed of the following linearized equations:\footnote{See Appendix \ref{sec:app:linearization} for all the derivations.
Following \citet{Leeper:1991}, the model is linearized rather than log-linearized. Hence, in this section hatted variables indicate deviations rather that log-deviations from the steady state, while variables without subscript indicate steady state values.} 
\begin{align}
\hat{R}_{t}     & =\frac{1}{\beta}E_{t}\hat{\pi}_{t+1},\label{eq:linLeep_Fisher}\\
\hat{m}_{t}     & =-\frac{c}{\left(  R-1\right)  ^{2}}\hat{R}_{t},\label{eq:linLeep_money}\\
\hat{b}_{t}+\hat{m}_{t}+\hat{\tau}_{t}+\frac{m+bR}{\pi^{2}}\hat{\pi}_{t}  &
=\frac{1}{\pi}\hat{m}_{t-1}+\frac{b}{\pi}\hat{R}_{t-1}+\frac{R}{\pi}\hat
{b}_{t-1}\label{eq:linLeep_budcon},\\
\hat{\tau}_{t}  &=\gamma\hat{b}_{t-1}+\psi_{t} \label{eq:linLeep_fiscal}.
\end{align}
Equation \eqref{eq:linLeep_Fisher} is a Fisher equation, relating the gross nominal interest rate to inflation, keeping the real rate fixed at its steady state value $\beta^{-1}$. Equations \eqref{eq:linLeep_money} and \eqref{eq:linLeep_budcon} are a money demand relation and the government budget constraint expressed in real terms, where $\hat{m}_{t}$ are real money balances, $\hat b_t$ is the real level of debt, $\hat \tau_t$ are real lump-sum taxes net of transfers, and $c$ represents real consumption expenditure.
Finally, equation \eqref{eq:linLeep_fiscal} is a fiscal policy rule: the government adjusts lump-sum taxes in response to lagged debt according to the parameter $\gamma$ and to a fiscal policy shock, $\psi_t$.

We generalize \citeauthorpos{Leeper:1991} monetary policy rule with
\begin{equation} \label{eq:linLeep_monpolrule}
\hat{R}_t=  \frac{\phi_p+\phi_\pi}{\pi}\hat{\pi}_t-\frac{\delta \phi_\pi}{\pi} \hat{\pi}_{t-1}+\delta \hat{R}_{t-1}+ \theta_t-\delta \theta_{t-1},
\end{equation}
where $\theta_t$ is a monetary policy shock.
Such a rule returns the traditional IT case, i.e., $\hat{R}_t=\frac{\phi_\pi}{\pi} \hat{\pi}_t+\theta_t$, when $\delta = 0$ and $\phi_p = 0$.  It reduces to the case of strict PLT (or superinertial) rule when $\delta=1$ and $\phi_\pi=0$, that is
\begin{equation} \label{eq:linLeep_PLTrule}
\hat{R}_t= \frac{\phi_p}{\pi}\hat{\pi}_t+ \hat{R}_{t-1} + \Delta\theta_t,
\end{equation}
which corresponds to equation \eqref{eq:wicks_rule} in the previous section.
Note that both $\psi_t$ and $\theta_t$ follow a stationary AR(1) with autoregressive parameter equal to $\rho_\psi$ and $\rho_\theta$, respectively.

Substituting the fiscal rule, the PLT rule, and the money demand equation into the government budget constraint we obtain an expression for real debt:
\begin{equation} \label{eq:linLeep_LOMdebt}
\varphi_1 \hat\pi_t + \hat b_t - \left(\frac{1}{\beta} -\gamma\right) \hat b_{t-1}+\varphi_2\hat R_{t-1}+\varphi_3\Delta\theta_t+\psi_t=0,
\end{equation}
where the coefficients $\varphi_1$, $\varphi_2$ and $\varphi_3$ are reported in Table \ref{tab:solutions}. 

Therefore, to derive both the determinacy conditions and the solution of the rational expectation equilibrium when the central bank adopts a PLT approach, we consider the system formed by equations \eqref{eq:linLeep_Fisher},  \eqref{eq:linLeep_PLTrule} and \eqref{eq:linLeep_LOMdebt}.

\subsection{Determinacy analysis}


Under passive fiscal policy, i.e., when $\vert \beta^{-1}-\gamma\vert <1$, determinacy requires $\phi_p$ to be positive. This parametric region defines the AM/PF regime, so active monetary policy corresponds to $\phi_p>0$.
Conversely, under active fiscal policy, i.e., when $\vert \beta^{-1}-\gamma\vert >1$, determinacy is achieved if $\phi_p$ is negative. This parametric region defines the PM/AF regime, so passive monetary policy corresponds to $\phi_p<0$.

Similarly to the analysis of \citet{Leeper:1991} under IT, when both authorities behave passively the system returns indeterminacy, while in case of jointly active fiscal and monetary policies the system has no stable solutions. Note that, differently from the standard IT case, under PLT monetary policy is active when the inflation coefficient in the monetary rule (\ref{eq:linLeep_PLTrule}) is positive and it is passive when that coefficient is negative. We defer further comments on this to the following section where we analyse a sticky price model.

\subsection{Model solution}
We use the method by \citet{Bhatta2014} to find the rational expectation solution for inflation under PLT (see Appendix \ref{app: sec: leeper}). Table \ref{tab:solutions} compares the solutions under PLT and IT in the two determinate regimes.
\begin{table}[ht]
\refstepcounter{table} \label{tab:solutions}
\centering
\begin{small}
\tabulinesep=4pt
\begin{tabu} to 1\textwidth {  lcX[l]  }
\multicolumn{3}{l}{Table \ref{tab:solutions}. Rational expectation solutions for inflation} \\
\toprule
\multirow{2}{*}{\textbf{AM/PF}} 
                    & PLT & 
                    $ \hat{\pi}_{t} =-\frac{\pi}{\phi_{p}}\hat
{R}_{t-1}-\frac{\beta}{1+\frac{\beta\phi_{p}}{\pi}-\rho_{\theta}}\theta_{t}+\frac{\pi}{\phi_{p}}\theta_{t-1}$
                   \\ 
                    & IT &  $\hat{\pi}_{t}   =\frac{\beta}{\rho_{\theta}%
-\frac{\beta\phi_{\pi}}{\pi}}\theta_{t}$ \\
                    \noalign{\medskip}
\multirow{2}{*}{\textbf{PM/AF}} 
                    & PLT & $\hat{\pi}_{t}  =-\frac{J}{K}\hat{R}%
_{t-1}-\left(  \frac{1}{\beta}-\gamma\right)  \frac{H}{K}\hat{b}_{t-1}%
+\frac{\left(  \frac{1}{\beta}-\gamma\right)  \frac{H}{K}}{\frac{1}{\beta
}-\gamma-\rho_{\psi}}\psi_{t}+\frac{\left(  \frac{1}{\beta}-\gamma\right)
\frac{H}{K}\varphi_{3}-\frac{J}{K}-\beta}{\frac{1}{\beta}-\gamma}\left(
\frac{\frac{1}{\beta}-\gamma-1}{\frac{1}{\beta}-\gamma-\rho_{\theta}}%
\theta_{t}-\theta_{t-1}\right)$ \\
                    & IT        
                    &  $\hat{\pi}_{t}   =\tilde{\varphi}_{2}\frac
{\tilde{H}}{\tilde{K}}\hat{R}_{t-1}-\left(  \frac{1}{\beta}-\gamma\right)
\frac{\tilde{H}}{\tilde{K}}\hat{b}_{t-1}+\frac{\left(  \frac{1}{\beta}%
-\gamma\right)  \frac{\tilde{H}}{\tilde{K}}}{\left(  \frac{1}{\beta}%
-\gamma\right)  -\rho_{\psi}}\psi_{t}+\frac{\left(  \frac{1}{\beta}%
-\gamma\right)  \frac{\tilde{H}}{\tilde{K}}\tilde{\varphi}_{3}+\tilde{\varphi
}_{2}\frac{\tilde{H}}{\tilde{K}}-\beta}{\left(  \frac{1}{\beta}-\gamma\right)
-\rho_{\theta}}\theta_{t} $ \\ 
\midrule
\end{tabu}
\tabulinesep=5pt
\begin{tabu} to 1\textwidth {  X[l]X[l]X[l]  }
\textit{Coefficients} \\
$\varphi_{1} =\frac{c}{R-1}\frac{1}{\pi}\left(  \frac{1}{\beta}-\frac{\phi_{p}}{R-1}\right)  +\frac{b}{\beta\pi}$	&
$\varphi_{2} =\frac{1}{\pi}\frac{c}{\left(  R-1\right)  ^{2}}-\frac{c}{\left(  R-1\right)  ^{2}}-\frac{b}{\pi}$	&
$\varphi_{3} =-\frac{c}{\left(  R-1\right)^{2}}$	\\
$\tilde{\varphi}_{1} =\frac{c}{R-1}\frac{1}{\pi}\left(\frac{1}{\beta}-\frac{\phi_{\pi}}{R-1}\right)+\frac{b}{\beta\pi}$	&
$\tilde{\varphi}_{2} =\frac{1}{\pi}\frac{c}{\left(  R-1\right)  ^{2}}-\frac{b}{\pi}$	&
$\tilde{\varphi}_{3} = {\varphi}_{3}$ \\
$H =1-\frac{1}{\beta}+\gamma+\frac{\beta\phi_{p}}{\pi}$	&
$K =\left(  \frac{1}{\beta}-\gamma-1\right)  \varphi_{1}+\frac{\phi_{p}}{\pi}\varphi_{2}$		&
$J =\beta\varphi_{1}+\left(  \frac{1}{\beta}-\gamma-\frac{\beta\phi_{p}}{\pi}\right)  \varphi_{2}$	\\
$\tilde{H}  =-\frac{1}{\beta}+\gamma+\frac{\beta\phi_{\pi}}{\pi}$ &
$\tilde{K}  =\left(  \frac{1}{\beta}-\gamma\right)  \tilde{\varphi}_{1}+\frac{\phi_{\pi}}{\pi}\tilde{\varphi}_{2}$\\
\bottomrule
\end{tabu}
\end{small}
\end{table}

As under IT, the solution for inflation under PLT depends just on monetary shocks in an AM/PF regime, while it depends also on government debt and fiscal shocks in a PM/AF regime. Differently from the IT case, however, a super-inertial interest rate rule implies that the nominal interest rate is an endogenous state variable. Therefore, inflation depends on the lagged interest rate ($\hat R_{t-1}$) under both regimes.
In the AM/PF regime we have that $\phi_p>0$, so the coefficient relating past interest rate to current inflation ($-\frac{\pi}{\phi_p}$) is negative.
The intuition goes as follows. Imagine that $\hat \pi_{t-1}$ increases, so that the central bank increases the interest rate $\hat R_{t-1}$. However, given that PLT introduces history dependence in monetary policy, the policymaker should compensate the positive shock to inflation occurred in the past by reducing current inflation (bygones are not bygones).
As a consequence, $\hat  \pi_t$ will be negatively affected by a surge in $\hat R_{t-1}$.
In the PM/AF regime, inflation also depends on lagged debt (positively, as it can be shown) while the coefficient on past interest rate is more involved. Its sign will depend on the chosen calibration.  

\section{A simple New-Keynesian model with PLT}

We now depart from \citetpos{Leeper:1991} flexible price model and illustrate the key properties of PLT adapting a basic New-Keynesian model. We augment the model with a fiscal block and replace the traditional IT rule with a that nests both IT and PLT.  The resulting model is composed of four equations:\footnote{In this section, hatted variables indicate log-deviations from the steady state. See \citet{AFG2020} for the linearization of the government budget constraint.}
\begin{align}
\hat{y}_t   &= E_t \hat y_{t+1}-\left(\hat R_t-E_t \hat\pi_{t+1}\right)+ (1−\rho_\varepsilon)\varepsilon_t ,  \label{eq:DSGE_Euler},\\
\hat \pi_t   &= \beta E_t \hat\pi_{t+1}+\kappa \hat y_t    \label{eq:DSGE_NKPC}, \\
\hat b_{t} &= \frac{1}{\beta}\left(1-\frac{\tau}{b}\gamma\right)\hat b_{t-1}+\hat R_{t}-\frac{1}{\beta}\hat \pi_{t}-\frac{1}{\beta}\frac{\tau}{b}\psi_t\label{eq:DSGE_govbudgetconstr},\\
\hat R_{t} &= \left(\phi_p+\phi_\pi\right)\hat \pi_{t}-\delta \phi_\pi \hat \pi_{t-1}+\delta \hat R_{t-1}  + \theta_{t}-\delta\theta_{t-1}. \label{eq:DSGE_TR}
\end{align}
The model contains three exogenous variables: the demand shock $\varepsilon_t$, the monetary policy shock $\theta_t$, and the fiscal policy shock $\psi_t$.\footnote{$\varepsilon_t$ represents the innovation to a preference shock on the intertemporal discount rate, modelled as in \citet{LiuWaggonerZha2009}. The preference shock, the monetary policy shock and the fiscal policy shock evolve according to stationary AR(1) processes.}

Section \ref{sec:det} contains the determinacy analysis of this simple NK model for which we provide analytical results. However, it is not possible to solve the model analytically, so Section \ref{sec:irf} present the simulated impulse responses to a demand and a fiscal shock under both the monetary and the fiscally-led regime. Section \ref{sec: welfare} compares the welfare losses under IT and PLT rules after a demand shock, while Section \ref{sec: ZLB} includes impulse response following the same demand shock in the presence of a zero lower bound.

\subsection{Determinacy analysis}\label{sec:det}

The determinacy properties of this simple DSGE model trace the findings obtained previously in the flexible price model.

\begin{proposition} Determinacy requires either: 
\begin{itemize}
    \item an AM/PF mix such that $\gamma<\frac{b}{\tau }(1+\beta )\text{ and }\phi_p>0$,
    \item or a PM/AF mix  such that $\gamma>\frac{b}{\tau }(1+\beta )\text{ and }\phi_p<0$.
\end{itemize}
\end{proposition}
\begin{proof}
See Appendix \ref{app: sec: DSGE}.
\end{proof}


\citet{Giannoni2014} undertakes the same analysis considering just the case of a passive fiscal policy. We confirm his finding that under PLT the determinacy condition for monetary policy is less restrictive than the Taylor principle. Interest rate reaction to an increase in inflation should be positive but, in case of price-level stabilization, it could even be lower than one.

The new result here pertains to the active fiscal policy case. There is determinacy whenever the nominal interest rate moves in opposite direction with respect to inflation $\left( \phi_p<0\right) $: if inflation decreases, the central bank must increase the nominal interest rate so that the real interest rate increases unambiguously.

\citet{Uribe2018} suggests, to exit a liquidity trap, to permanently raise nominal interest rates as this would increase inflation, according to the so-called neo-Fisherian effect. Interestingly, in case of adoption of a PLT combined with a PM/AF regime, the same would happen: whenever inflation decreases the central bank must increase the interest rate. As explained below, however, the economic mechanism is very different here and it is due to the wealth effect induced by the fiscal theory of the price level.

\subsection{Impulse responses}\label{sec:irf}
We now want to study the dynamics of economic variables when hit by a demand shock and by a fiscal shock. We analyse both the monetary-led regime (with $\gamma=0.2$) and the fiscally-led one (with $\gamma=0$) and we compare
impulse response functions under IT (when $\delta =0$ and $\phi_p=0$ in (\ref{eq:DSGE_TR}) to those under PLT (when $\delta =1,$ $\phi_\pi=0$). Results in the literature for the PLT case are limited to the study of the monetary-led regime, while here we enrich the explanation of the dynamics in the fiscally-led regime.


\subsubsection{Negative Demand shock}

Consider a negative demand shock,
as shown in Figure \ref{fig:shock_demand}.

\begin{figure}[h!]
\begin{subfigure}{.5\textwidth}
\centering
  \includegraphics[trim = 10mm 13mm 75mm 10mm,clip,scale=1]{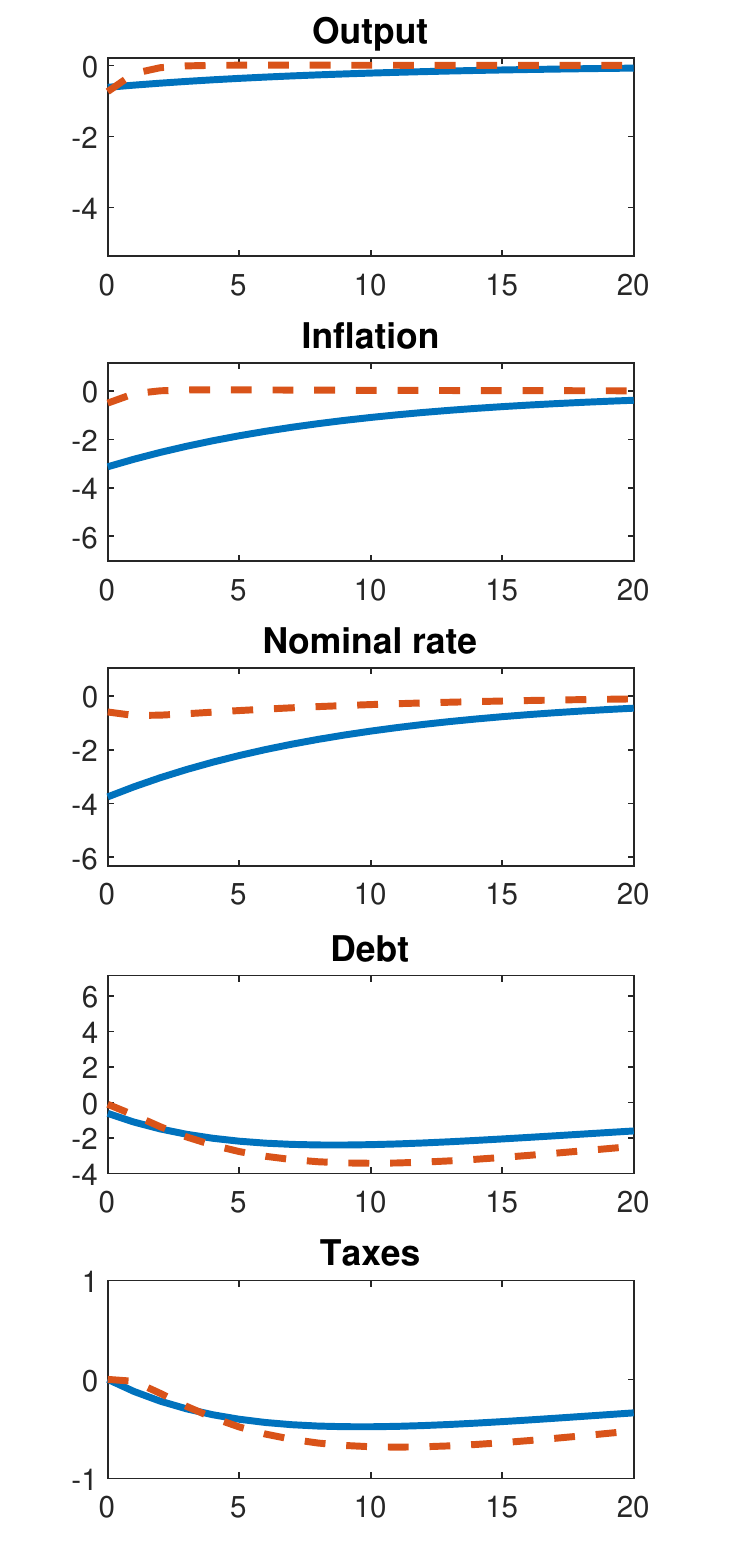}
  \caption{AM/PF mix}
  \label{fig: demandAMPF}
\end{subfigure}%
\begin{subfigure}{.5\textwidth}
  \centering
  \includegraphics[trim = 85mm 13mm 10mm 10mm,clip,scale=1]{figs/shock_a.pdf}
  \caption{PM/AF mix}
   \label{fig: demandPMAF}
\end{subfigure}
\centering
\begin{tabu} to 0.9\textwidth {X[l]}
\caption{Impulse response function to a negative demand shock.
\newline
{\protect\footnotesize \textit{Notes:} Solid blue lines: inflation targeting; dashed red lines: price level targeting.
Parametrization of the AM/PF mix: $\phi_\pi=\phi_p=1.2$, $\gamma=0.2$. 
Parametrization of the PM/AF mix: $\phi_\pi=0.9$, $\phi_p=-0.1$, $\gamma=0$.} }
\label{fig:shock_demand}
\end{tabu}
\end{figure}

\textbf{AM/PF.} Figure \ref{fig: demandAMPF} exhibits the well-known results under a AM/PF regime. A negative demand shock reduces inflation, output and the nominal interest rate, and by a greater extent under IT. Therefore, in an AM/PF regime, there are advantages from the adoption of a PLT approach. 

\textbf{PM/AF.} Figure \ref{fig: demandPMAF} shows that a negative demand shock decreases inflation less under PLT. Output decreases more on impact, but it exhibits a less persistent dynamics, returning to equilibrium after 5 periods.
To have an intuition of what is going on, consider the government debt
valuation equation:
\begin{equation} \label{eq:debt_valuation}
\frac{B_{t-1}}{P_t}=\sum\limits_{j=0}^{\infty }\frac{s_{t+j}}{r^{j}}, 
\end{equation}
where $B_{t-1}$ is nominal debt (predetermined), $P_t$ is the price level, $s$ is government surplus and $r$ is the real interest rate (from the Fisher equation: $r_t=R_t-E_t\pi_{t+1}$). The real value of government debt (LHS) must be backed by the present value of future primary surpluses (RHS). Active fiscal policy implies that unbacked fiscal expansions - an increase in government expenditures not combined to a tax rise to cover the deficit - induce wealth effects: agents will substitute out of debt holdings and this would raise consumption demand and increase current prices to a level that restores balance between the LHS and the RHS.
Under IT, output decreases and so does inflation creating a negative
inflation tax on government's nominal liabilities which increases real
government debt (LHS increases). Since fiscal policy does not adjust taxes or public expenditures, the LHS is larger than the RHS so wealth effects kick in and these make inflation increase (inflation reversal) and output too. The inflation increase helps to keep the government budget constrain satisfied.

Output and inflation decrease under PLT too. However, as a response to the inflation reduction, the central bank increases the nominal interest rate, $R$, under PLT and this, in turn, raises the real interest rate, $r$. This has two consequences. First, there is a much larger reduction in output on impact. Second, a rise in  $r$ enlarges the difference between the LHS (which increases on impact) and the RHS of (\ref{eq:debt_valuation}), because $r$ is at the denominator on the RHS. Hence, \textit{wealth effects} are now larger, inducing a much quicker rebound of both inflation and output.


After a decrease in inflation the central bank raises the interest rate and this, in turn, creates wealth effects that spur inflation. As previously stressed, here the recipe to go out from a deflationary trap is similar to the one proposed exploiting the neo-Fisherian effect: increase interest rates to increase inflation. However, here the logic is very different from the neo-Fisherian perspective, and it is based on the wealth effects due to the fiscal theory of the price level.

Furthermore, note that, contrary to what happens under IT, if the central bank is adopting a PLT approach, such a shock does not make the nominal interest rate decrease. In this case, PLT also serves the purpose of preventing the economy from hitting the zero lower bound.

\subsubsection{Fiscal shock}

Figure \ref{fig:shock_fiscal} shows the effects of an expansionary fiscal shock brought about by a tax reduction.

\begin{figure}[h!]
\begin{subfigure}{.5\textwidth}
\centering
  \includegraphics[trim = 10mm 13mm 75mm 10mm,clip,scale=1]{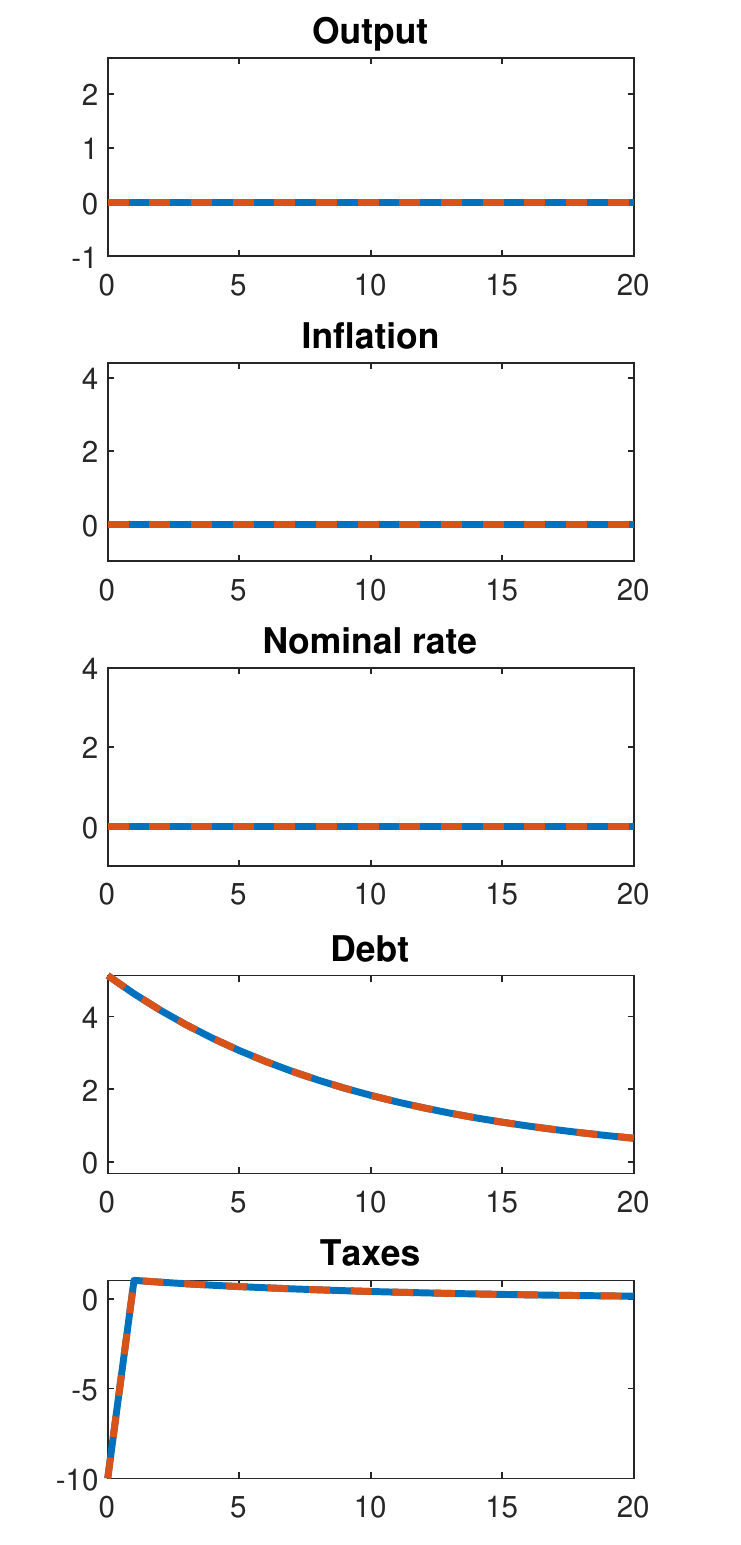}
  \caption{AM/PF mix}
  \label{fig:fiscAMPF}
\end{subfigure}%
\begin{subfigure}{.5\textwidth}
  \centering
  \includegraphics[trim = 85mm 13mm 10mm 10mm,clip,scale=1]{figs/shock_f.pdf}
  \caption{PM/AF mix}
  \label{fig:fiscPMAF}
\end{subfigure}
\centering
\begin{tabu} to 0.9\textwidth {X[l]}
\caption{Impulse response function to an expansionary fiscal shock.
\newline
{\protect\footnotesize \textit{Notes:} Solid blue lines: inflation targeting; dashed red lines: price level targeting. 
Parametrization of the AM/PF mix: $\phi_\pi=\phi_p=1.2$, $\gamma=0.2$. 
Parametrization of the PM/AF mix: $\phi_\pi=0.9$, $\phi_p=-0.1$, $\gamma=0$.} }
\label{fig:shock_fiscal}
\end{tabu}
\end{figure}

\textbf{AM/PF.} Because of Ricardian equivalence, a fiscal policy shock with lump-sum taxes in the AM/PF regime turns out to be ineffective on inflation, output and nominal interest rates (see Figure \ref{fig:fiscAMPF}). If the government decreases taxes, rational agents anticipate a future fiscal adjustment and so they save today to pay for higher taxes tomorrow: there are no wealth effects and economic variables do not move. This is true both under IT and under PLT.

\textbf{PM/AF.} Figure \ref{fig:fiscPMAF} shows the dynamics after an expansionary fiscal shock in a PM/AF regime under both IT and PLT.
Under IT, a tax reduction decreases the present discounted value of
surpluses (RHS of \eqref{eq:debt_valuation}). The government debt owned by households exceeds the present discounted value of surpluses (LHS $>$ RHS), and this represents a positive wealth effect. Agents anticipate that surpluses will not be covered by future fiscal adjustments and thus they convert bonds into current consumption goods, spending increases and inflation surges. This inflation increase is accommodated by the central bank that raises $R$ but less than one-for-one with the increase in inflation. This stimulates output and makes the debt over GDP ratio decrease.

Under PLT dynamics are different. As before, a tax reduction decreases the present discounted value of surpluses and the government debt owned by households exceeds the present discounted value of surpluses (LHS $>$ RHS) engineering a positive wealth effect: spending increases and inflation too. However, now, following the increase in the price level, the central bank decreases $R$. Thus, the real interest rate $r$ goes down by a greater extent than under IT. Since $r$ appears at the denominator on the RHS of \eqref{eq:debt_valuation}, the more it
decreases, the more it offsets the reduction in the fiscal surplus; therefore, the less the LHS will be higher than the RHS, and \textit{the lower the wealth effect} will be, the lower the increase in inflation and in the debt/GDP ratio. 
Furthermore, the largest decrease in $r$ stimulates output.

\subsection{Welfare analysis} \label{sec: welfare}

Following \citet{GorodShapiro2007}, we undertake a welfare analysis
to evaluate the performance of IT and PLT rules following a demand shock. In order to do so we employ the following loss function:
\begin{equation}
\mathcal{L}=\sum\limits_{t=0}^{\infty }\beta ^{t}(\omega_\pi \pi_t^{2}+\omega_x x_t^{2}+\omega_R(R_t-R_t^{\ast }))=\omega_\pi\mathcal{L}_\pi+\omega_x\mathcal{L}_x+\omega_R\mathcal{L}_{R},
\end{equation}
where $\omega_\pi$, $\omega_x$ and $\omega_R$ are the weights on inflation, the output gap and the deviation of the interest rate from its target path.
As \citet{GorodShapiro2007}, we bias the findings
against PLT by not including a possible term on the price-level gap and by associating a relatively large weight to the output gap ($\omega_x=1$), since PLT usually increases output volatility with respect to IT.

Figure \ref{fig:welfare} shows how the loss function varies as the inflation (or price level) coefficient in the monetary rule changes in the AM/PF and PM/AF regimes. The monetary policy coefficient goes from passive to active under IT (panel \ref{fig:welfarea}, cut-off value $\phi_\pi=1$) and under PLT (panel \ref{fig:welfareb}, cut-off value $\phi_p=0$). The fiscal coefficient ($\gamma$) is put equal to 0 (active fiscal policy) when monetary policy is passive, and to 0.2 (passive fiscal policy) when monetary policy is active. The yellow and red lines in the figure describe the loss function when the weights $\omega_\pi$ and $\omega_x$ are each, respectively, set equal to one (while the others are set to zero) and the total loss (in blue) when $\omega_\pi=$ $\omega_x=$ $\omega_R=1$ contemporaneously.\footnote{For the sake of clarity, in the figure we omit the loss function when $\omega_R=1$ and the others are set to zero but, for completeness, we report the total loss function when even this interest rate component is included.}

\begin{figure}[h!]
\begin{subfigure}{1\textwidth}
\centering
  \includegraphics[trim = 0mm 90mm 0mm 10mm,clip,width=.7\textwidth]{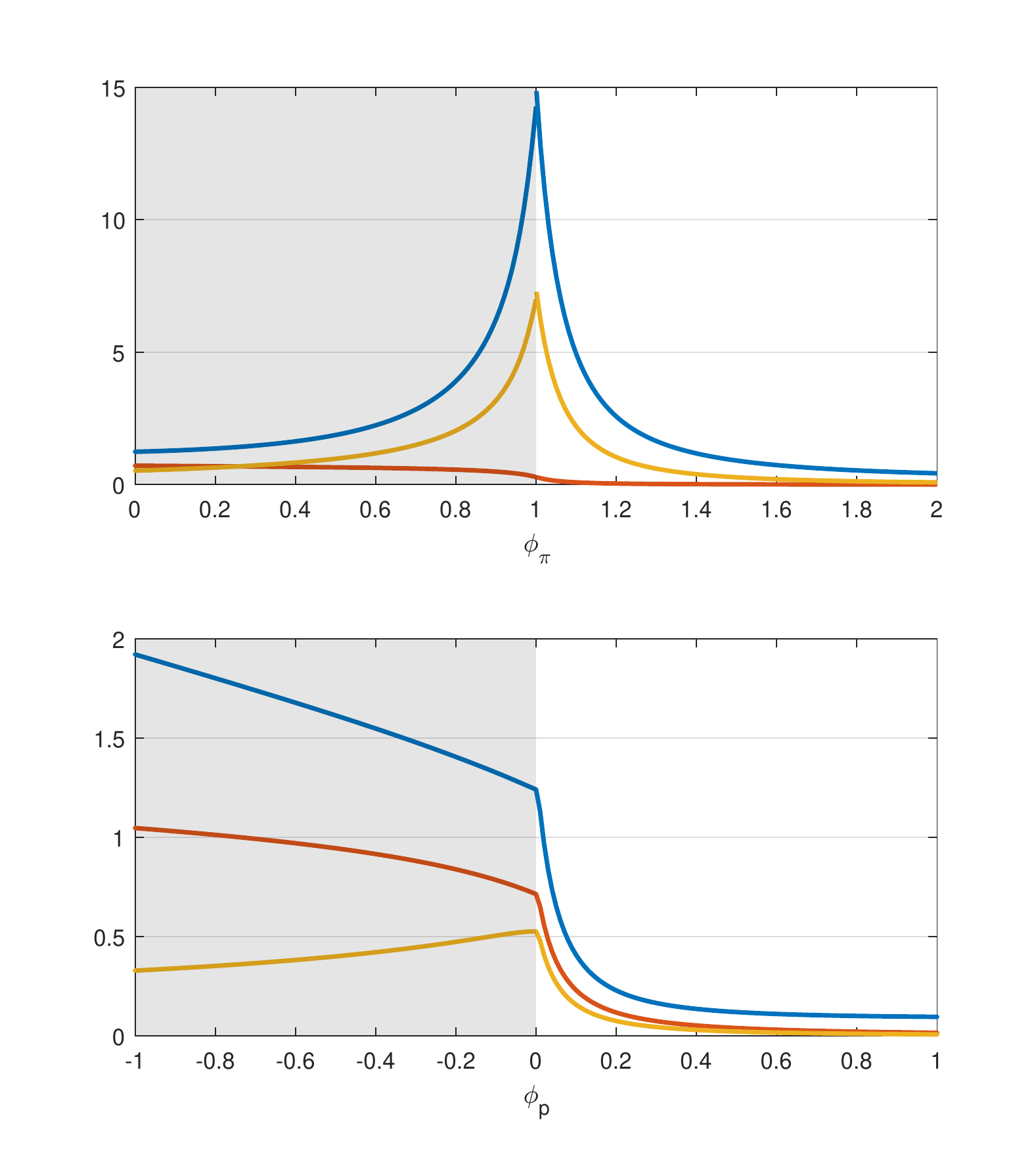}
  \caption{Inflation targeting}
  \label{fig:welfarea}
  \bigskip
\end{subfigure}
\begin{subfigure}{1\textwidth}
  \centering
  \includegraphics[trim = 0mm 05mm 0mm 95mm,clip,width=.7\textwidth]{figs/welfare_a.pdf}
  \caption{Price level targeting}
  \label{fig:welfareb}
\end{subfigure}
\centering
\begin{tabu} to 0.9\textwidth {X[l]}
\caption{Loss function after a negative demand shock.
\newline
{\protect\footnotesize \textit{Notes:} Shaded background: PM/AF mix; white background: AM/PF mix. Blue lines: total loss; red lines: output component; yellow lines: inflation component.} }
\label{fig:welfare}
\end{tabu}
\end{figure}

In the AM/PF case (white background), we find the well-known result, e.g., \citet{SGU2007}, that monetary policy can completely stabilize the output gap and inflation facing a demand shock by making the inflation coefficient (or the price level one) tend to infinity. This would be the overall preferred policy combination after a demand shock. Note that both the output gap and the inflation gap components decreases with $\phi_\pi$.

Here, we are more interested int he shaded area, that is, the PM/AF regime. In the IT case, the overall loss function (blue line) increases as it approaches the cut-off value of 1, so a very passive monetary policy would be preferred. As monetary policy becomes more active, the output gap volatility (red line) decreases, while the inflation gap one (yellow line) increases sharply. The latter effect dominates in determining the overall loss.
The opposite happens in the PLT case. The overall loss is decreasing in $\phi_p$, so that a less passive policy is preferred. While, as for the IT case, as monetary policy becomes more active, the output gap volatility decreases and the inflation gap one increases, under PLT it is the former effect to dominate in determining the overall loss. 

As expected, PLT determines a larger output gap volatility and a smaller inflation gap volatility than IT. However, the overall welfare loss function is much smaller under PLT, so that, conditional on a demand shock, PLT generally outperforms IT.

\subsection{The case of a zero lower bound} \label{sec: ZLB}

Following \citet{EggertsonWoodford2003}, we now assume a demand shock that hits the economy depressing output and inflation and we introduce a ZLB that prevents the central bank to provide the appropriate stimulus to offset the recession. Also in this case we consider both the IT and the PLT approaches under a monetary-led and a fiscal-led regime.
Under an AM/PF regime, Figure \ref{fig:shock_demand_ZLBa} shows that inflation and output decrease by a greater extent under IT, while they hardly move under PLT. Under IT the nominal interest rate hits the zero lower bound after two periods, while it remains unconstrained under PLT. Even the path of debt is more favorable under PLT. 
Similar comments apply to the case of a PM/AF regime; however, in this case, inflation and output decrease under PLT too (though much less than in the IT case, see Figure \ref{fig:shock_demand_ZLBb}). 
Therefore, even under the threat of reaching the zero bound, following a demand shock, PLT outperforms IT, in terms of inflation, output and debt, both under a monetary- and a fiscally-led regime, and it always avoids a liquidity trap. 


\begin{figure}[h!]
\begin{subfigure}{.5\textwidth}
\centering
  \includegraphics[trim = 0mm 13mm 0mm 0mm,clip,scale=1]{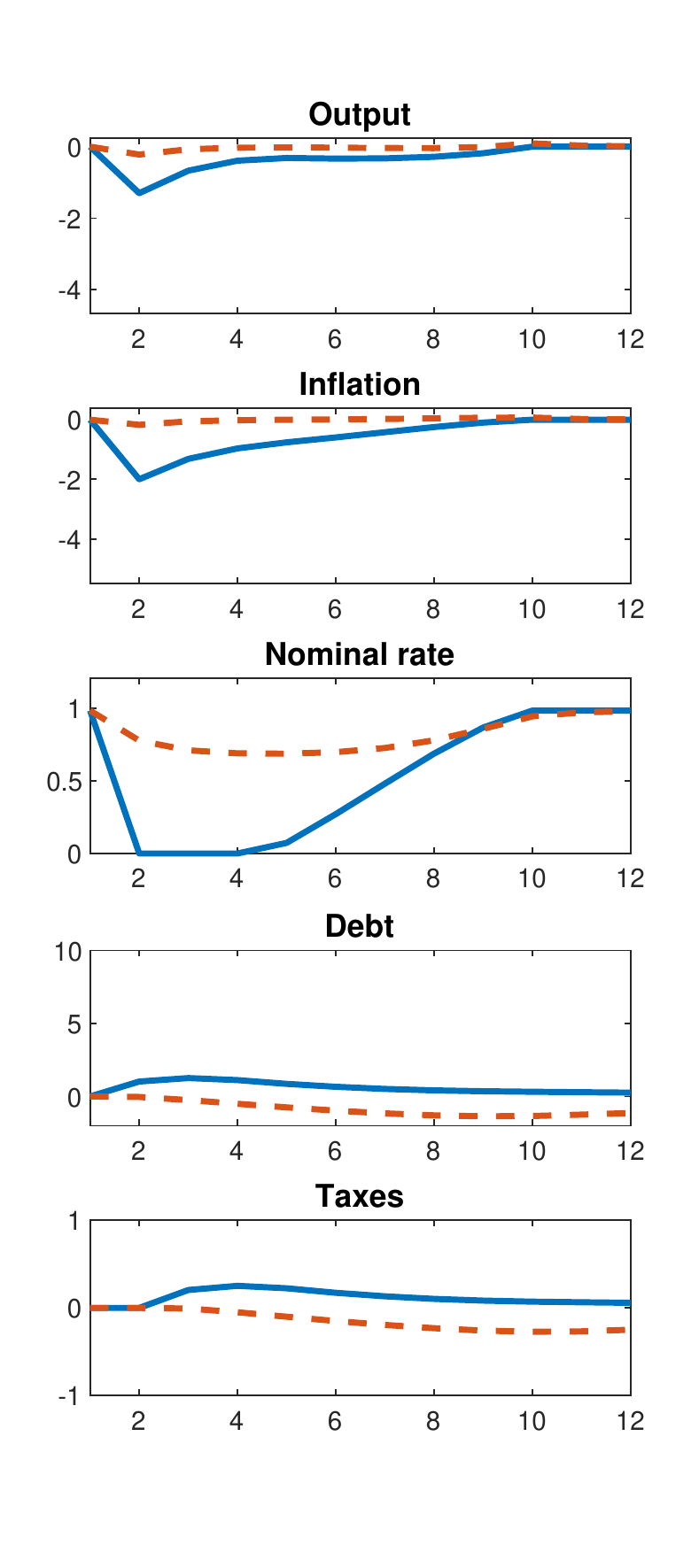}
  \caption{AM/PF mix}
  \label{fig:shock_demand_ZLBa}
\end{subfigure}%
\begin{subfigure}{.5\textwidth}
  \centering
  \includegraphics[trim = 0mm 13mm 0mm 0mm,clip,scale=1]{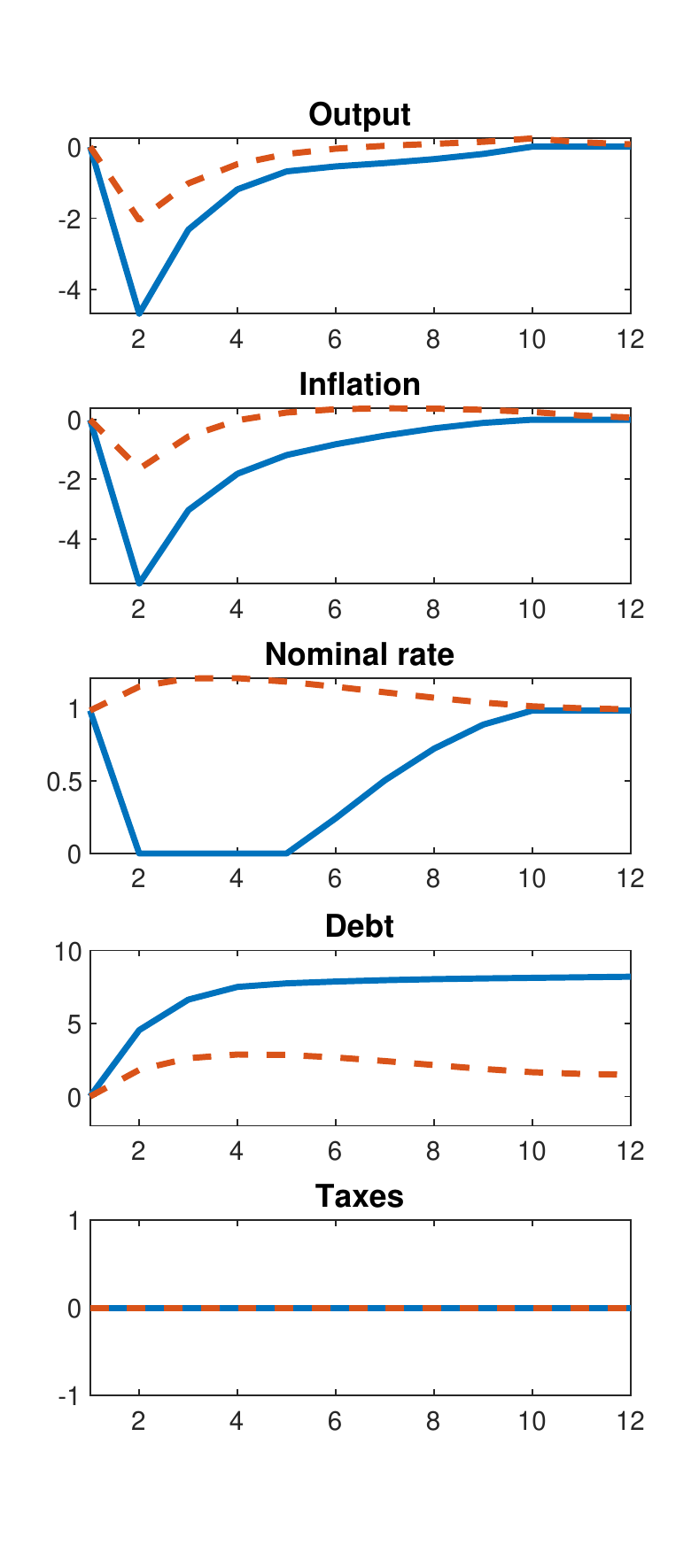}
  \caption{PM/AF mix}
  \label{fig:shock_demand_ZLBb}
\end{subfigure}
\centering
\begin{tabu} to 0.9\textwidth {X[l]}
\caption{Impulse response function to a negative demand shock.
\newline
{\protect\footnotesize \textit{Notes:} Solid blue lines: inflation targeting; Dashed red lines: price level targeting. 
Parametrization of the AM/PF mix: $\phi_\pi=\phi_p=1.2$, $\gamma=0.2$. 
Parametrization of the PM/AF mix: $\phi_\pi=0.9$, $\phi_p=-0.1$, $\gamma=0$.} }
\label{fig:shock_demand_ZLB}
\end{tabu}
\end{figure}

\section{Conclusions}

A negative demand shock has less severe consequences on the economy if the central bank is adopting PLT: inflation, output and the nominal interest rate decrease less than under IT. While a ZLB would be binding under IT, PLT allows the central bank to avoid it. We show that these results, well-known under the traditional monetary-led regime, hold under the less-studied fiscally-led regime too. This last case is of interest since it combines both the recent adoption of a makeup strategy and a more active fiscal policy needed to face deflationary risks. 

Moreover, in a fiscally-led regime with PLT monetary policy \textit{increases the nominal interest rate} following a deflationary demand shock and, by doing so, it avoids a prolonged period of low output and inflation. The nominal interest rate increases because, if fiscal policy is active, determinacy requires monetary policy to be passive. However, under PLT Leeper's condition is more restrictive: the inflation coefficient in the interest rate rule should be lower than zero, rather than one. Hence, the nominal interest rate increases if the price gap is negative. Whenever there is a demand shock, this ``inverse'' reaction of the policy rate to inflation exacerbates wealth effects; while, if a shock hits the government surplus (see the tax shock above), it dampens them. Analysing the social loss function, we find that PLT generally dominates IT even from a social welfare point of view.  

\citet{Bernanke2017} proposal of a temporary price-level targeting entails a switch from IT to PLT whenever the economy is about to approach the ZLB. According to this (preliminary) analysis, adopting a PLT approach could be beneficial, whatever the monetary/fiscal policy mix. Future work should address this topic employing a Markov-switching approach that analyses both the switch from IT to PLT and that from a monetary- to a fiscally-led regime. The Markov switching framework would, among the other things, stress the role of the expectation effects and, since expectations on future policy moves affect the present value of primary surplus, this would shade more light on wealth effects too.

\FloatBarrier
\newpage
\bibliographystyle{ecta}
\bibliography{biblio}

\appendix
\renewcommand{\theequation}{\thesection\arabic{equation}}
\setcounter{equation}{0}

\section{Analytical results for \citetpos{Leeper:1991} model}
\subsection{Linearization of the  model
\label{sec:app:linearization}}
The original model by \citet{Leeper:1991} is formed by the following equations: a Fisher equation,
\begin{equation} \label{eq:Leeper_fisher}
R_t=\frac{1}{\beta}E_t\pi_{t+1}  ,
\end{equation}
where $\pi_t = P_t/P_{t-1}$ and $R_t$ is the gross nominal interest rate; a money demand relation,
\begin{equation} \label{eq:Leeper_moneydemand}
m_t=c\frac{R_t}{R_t-1},
\end{equation}
where $m_t$ are real money balances;
a government budget constraint,
\begin{equation} \label{eq:Leeper_govbudgetconstr}
b_t+m_t+\tau_t=g+\frac{m_{t-1}}{\pi_t}%
+R_{t-1}\frac{b_{t-1}}{\pi_t}   ,
\end{equation}%
where $b_t$ is the level of real government debt, consisting in one-period bonds,
$\tau_t$ are real lump-sum taxes net of transfers, while $c$ and $g$ are real consumption and government expenditures, which are assumed to be constant. The model is closed by two feedback rules: one for fiscal policy and one for monetary policy. 
Lump-sum taxes react to lagged debt
\begin{equation} \label{eq:Leeper_taxrule}
    \tau_t=\gamma_0+\gamma b_{t-1}+\psi_t,
\end{equation}
where $\psi_t$ is a fiscal policy shock. With respect to \citet{Leeper:1991}, we generalize the monetary policy rule letting the central bank react both to the deviations of the price level and of the inflation rate from their targets 
\begin{equation}\label{eq:Leeper_monpolrule2}
R_t=\phi_0+\phi_p\tilde P_t+\phi_\pi \tilde  \pi_t+\theta_t,
\end{equation}
where $\tilde P_t = P_t/P^*_t$ and $\tilde\pi_t=\pi_t/\pi^*$. Note that the steady state level of inflation is equal to the inflation target, so we can substitute $\pi^*=\pi$.
Following \citet{GorodShapiro2007}, we define the law of motion of the target price level as 
$$
P_{t}^{\ast}  =P_{t-1}^{\ast}\left(  \frac{P_{t-1}}{P_{t-1}^{\ast}}\right)
^{1-\delta}\pi.
$$
Using the identity $P_t = \pi_t P_{t-1}$, we obtain that the price level gap evolves as 
\begin{equation}\label{eq:price_deviations}
\tilde P_t=\tilde P_{t-1}^\delta \tilde\pi_t.
\end{equation}

To get a linearized version of the model, let's first consider the monetary policy rule. From \eqref{eq:Leeper_monpolrule2} we obtain 
\begin{equation}\label{eq:Leeper_monpolrule_lin}
\hat{R}_{t}  =\phi_{p}\hat{P}_{t}+\frac{\phi_{\pi}}{\pi}\hat{\pi}_{t}+\theta_{t} ,
\end{equation}
where hatted variables are defined $\hat{R}_{t}= R_{t}-R$ and the likes, while from \eqref{eq:price_deviations} we get
\begin{equation}\label{eq:price_deviations_lin}
\hat{P}_{t} =\delta\hat{P}_{t-1}+\frac{\hat{\pi}_{t}}{\pi}.
\end{equation}
After quasi-differencing equation \eqref{eq:Leeper_monpolrule_lin} and using \eqref{eq:price_deviations_lin} to substitute the price deviations, we obtain the expression 
\begin{equation}
\hat{R}_t=  \frac{\phi_p+\phi_\pi}{\pi}\hat{\pi}_t-\frac{\delta \phi_\pi}{\pi} \hat{\pi}_{t-1}+\delta \hat{R}_{t-1}+ \theta_t-\delta \theta_{t-1},\label{eq:Leeper_monpolrule}
\end{equation}
which corresponds to equation \eqref{eq:linLeep_monpolrule} in the text.
Note that for $\delta=0$ and $\phi_p = 0$  we are back to the usual case of IT, while for  $\delta=1$ and $\phi_\pi = 0$ we obtain the strict PLT rule. 

The other equations of the model, that is, equations \eqref{eq:Leeper_fisher}-\eqref{eq:Leeper_taxrule}, can be linearized to obtain equations \eqref{eq:linLeep_Fisher}-\eqref{eq:linLeep_fiscal} in the text, which we report here for convenience:
\begin{align*}
\hat{R}_{t}     & =\frac{1}{\beta}E_{t}\hat{\pi}_{t+1},\\
\hat{m}_{t}     & =-\frac{c}{\left(  R-1\right)  ^{2}}\hat{R}_{t},\\
\hat{\tau}_{t}  & =\gamma\hat{b}_{t-1}+\psi_{t},\\
\hat{b}_{t}+\hat{m}_{t}+\hat{\tau}_{t}+\frac{m+bR}{\pi^{2}}\hat{\pi}_{t}  &
=\frac{1}{\pi}\hat{m}_{t-1}+\frac{b}{\pi}\hat{R}_{t-1}+\frac{R}{\pi}\hat
{b}_{t-1}.
\end{align*}
We can substitute the linearized money demand relation, the fiscal rule and the monetary policy rule \eqref{eq:Leeper_monpolrule} inside the government budget constraint to obtain
\begin{multline*}
\left(\frac{m+bR}{\Pi^{2}}-\frac{c}{\left(  R-1\right)  ^{2}}
\frac{\phi_{p}+\phi_{\pi}}{\Pi}\right)  \hat{\Pi}_{t}
    +\hat{b}_{t}
    +\frac{c\delta}{\left(R-1\right)^{2}}\frac{\delta\phi_{\pi}}{\Pi}\hat{\Pi}_{t-1}
    -\left(\frac{1}{\beta}-\gamma\right)  \hat{b}_{t-1}\\
    +\left(\frac{1}{\Pi}\frac{c}{\left(R-1\right)^{2}}-\frac{1}{\Pi}b-\frac{c\delta}{\left(R-1\right)^{2}}\right)  \hat{R}_{t-1}
    +\psi_{t}
    -\frac{c}{\left(R-1\right)^{2}}\left(\theta_{t}-\delta\theta_{t-1}\right)=0.
\end{multline*}
If we now impose the parametrization for the PLT rule ($\phi_\pi=0$ and $\delta=1$) we get
$$\varphi_1 \hat\pi_t + \hat b_t - \left(\frac{1}{\beta} -\gamma\right) \hat b_{t-1}+\varphi_2\hat R_{t-1}+\varphi_3\Delta\theta_t+\psi_t=0,$$
which corresponds to equation \eqref{eq:linLeep_LOMdebt} in the text.
Instead, if we use the parametrization for the IT rule ($\phi_p=0$ and $\delta=0$) we obtain 
$$
\tilde{\varphi}_{1}\hat{\pi}_{t}+\hat{b}_{t}- \left(\frac{1}{\beta} -\gamma\right)  \hat{b}_{t-1}+\tilde{\varphi}_{2}\hat{R}_{t-1}+\tilde{\varphi}_{3}\theta_{t}+\psi_{t}=0.
$$
Parameters $\varphi_1$, $\varphi_2$, $\varphi_3$, $\tilde\varphi_1$, $\tilde\varphi_2$, and $\tilde\varphi_3$ are all reported in Table \ref{tab:solutions} in the text.

\subsection{Determinacy analysis under PLT \label{app: sec: leeper}}

We study the determinacy properties of the system formed by equations \eqref{eq:linLeep_Fisher},  \eqref{eq:linLeep_PLTrule} and \eqref{eq:linLeep_LOMdebt} in the text,
which we report again for convenience:
\begin{align*}
\hat{R}_t &= \frac{1}{\beta}E_t\hat{\pi}_{t+1}, \\
\hat{R}_t &= \frac{\phi_p}{\pi}\hat{\pi}_t+ \hat{R}_{t-1} + \Delta\theta_t, \\
0&=\varphi_1 \hat\pi_t + \hat b_t - \left(\frac{1}{\beta} -\gamma\right) \hat b_{t-1}+\varphi_2\hat R_{t-1}+\varphi_3\Delta\theta_t+\psi_t.
\end{align*}
In matrix form this system becomes
$$
\begin{bmatrix}
-\frac{1}{\beta} & 1 & 0\\
0 & 1 & 0\\
0 & 0 & 1
\end{bmatrix}%
\begin{bmatrix}
E_{t}\hat\pi_{t+1}\\
\hat R_{t}\\
\hat b_{t}%
\end{bmatrix}
=%
\begin{bmatrix}
0 & 0 & 0\\
\frac{\phi_{p}}{\pi} & 1 & 0\\
-\varphi_{1} & -\varphi_{2} & \frac{1}{\beta}-\gamma
\end{bmatrix}
\begin{bmatrix}
\hat\pi_{t}\\
\hat R_{t-1}\\
\hat b_{t-1}
\end{bmatrix}
+
\begin{bmatrix}
0 & 0\\
1 & 0\\
-\varphi_{3} & -1
\end{bmatrix}
\begin{bmatrix}
\Delta\theta_{t}\\
\psi_{t}%
\end{bmatrix}
$$
The matrix on the left hand side can be inverted to obtain
\begin{equation}\label{eq:matsystem_Leeper}
\begin{bmatrix}
E_{t}\hat\pi_{t+1}\\
\hat R_{t}\\
\hat b_{t}%
\end{bmatrix}
=%
\begin{bmatrix}
\frac{\beta\phi_{p}}{\pi} & \beta & 0\\
\frac{\phi_{p}}{\pi} & 1 & 0\\
-\varphi_{1} & -\varphi_{2} & \frac{1}{\beta}-\gamma
\end{bmatrix}
\begin{bmatrix}
\hat\pi_{t}\\
\hat R_{t-1}\\
\hat b_{t-1}
\end{bmatrix}
+
\begin{bmatrix}
\beta & 0\\
1 & 0\\
-\varphi_{3} & -1
\end{bmatrix}
\begin{bmatrix}
\Delta\theta_{t}\\
\psi_{t}%
\end{bmatrix}
\end{equation}
or
$$
E_t X_{t+1}=AX_{t}+B \epsilon_t 
$$
To satisfy the Blanchard-Kahn conditions for determinacy, the
matrix $A$ must have two eigenvalues inside and one outside the unit circle.
The eigenvalues of $A$ are $0$, $\frac{1}{\beta}-\gamma$ and $1+\frac{\beta\phi_{p}}{\pi}$. 
Therefore, if $\frac{1}{\beta}-\gamma$ is inside the unit
circle, that is if we have a passive fiscal policy, the other two
eigenvalues should be one inside and the other outside.  This happens when
$\phi_{p}>0$, which denotes an active monetary policy. 
If $\frac{1}{\beta}-\gamma$ is outside the unit circle, that is if we have an
active fiscal policy, the other two eigenvalues should be both inside.
This happens when $\phi_{p}<0$, that denotes a passive monetary policy.
Therefore, summarizing, determinacy requires either:
\begin{itemize}
    \item $\phi_{p}>0$ and $\left\vert \frac{1}{\beta}-\gamma\right\vert <1,$ the AM/PF
case, or
\item $\phi_{p}<0$ and $\left\vert \frac{1}{\beta}-\gamma\right\vert >1,$ the PM/AF case.
\end{itemize}

\subsection{Rational expectations solutions under PLT}

We now solve our model to find the solution for inflation both in the
monetary-led (AM/PF) regime and in the fiscally-led (PM/AF) regime. To find the solutions, we follow the procedure of \citet{Bhatta2014}, which is  based on the spectral decomposition of matrix $A =  VDV^{-1}$, where $D$ and $V$ are the matrices with the eigenvalues and the eigenvectors of $A$.
In particular, we obtain
\begin{align*}
D &= 
\begin{bmatrix}
0 & 0 & 0 \\ 
0 & 1+\frac{\beta \phi _{p}}{\pi } & 0 \\ 
0 & 0 & \frac{1}{\beta }-\gamma 
\end{bmatrix} 
= 
\begin{bmatrix}
    e_{1} & 0 & 0\\
    0 & e_{2} & 0\\
    0 & 0 & e_{3}%
\end{bmatrix} \\[1.5ex]
V^{-1} &= 
\begin{bmatrix}
\frac{1}{1+\frac{\beta \phi _{p}}{\pi }}\frac{\varphi _{1}-\frac{\phi _{p}}{%
\pi }\varphi _{2}}{\frac{1}{\beta }-\gamma } & -\frac{\beta }{1+\frac{\beta
\phi _{p}}{\pi }}\frac{\varphi _{1}-\frac{\phi _{p}}{\pi }\varphi _{2}}{%
\frac{1}{\beta }-\gamma } & 0 \\[1.5ex] 
-\frac{\frac{\phi _{p}}{\pi }}{1+\frac{\beta \phi _{p}}{\pi }}\frac{\beta
\varphi _{1}+\varphi _{2}}{1-\frac{1}{\beta }+\gamma +\frac{\beta \phi _{p}}{%
\pi }} & -\frac{1}{1+\frac{\beta \phi _{p}}{\pi }}\frac{\beta \varphi
_{1}+\varphi _{2}}{1-\frac{1}{\beta }+\gamma +\frac{\beta \phi _{p}}{\pi }}
& 0 \\[1.5ex]
\frac{1}{\frac{1}{\beta }-\gamma }\frac{\left( \frac{1}{\beta }-\gamma
-1\right) \varphi _{1}+\frac{\phi _{p}}{\pi }\varphi _{2}}{1-\frac{1}{\beta }%
+\gamma +\frac{\beta \phi _{p}}{\pi }} & \frac{1}{\frac{1}{\beta }-\gamma }%
\frac{\beta \varphi _{1}+\left( \frac{1}{\beta }-\gamma -\frac{\beta \phi
_{p}}{\pi }\right) \varphi _{2}}{1-\frac{1}{\beta }+\gamma +\frac{\beta \phi
_{p}}{\pi }} & 1%
\end{bmatrix}
= 
\begin{bmatrix}
q_{11} & q_{12} & 0\\
q_{21} & q_{22} & 0\\
q_{31} & q_{32} & 1
\end{bmatrix}.
\end{align*}
By multiplying both sides of \eqref{eq:matsystem_Leeper} by $V^{-1}$, we can diagonalize the system as 
\begin{equation}\label{eq:matsystem_Leeper_diag}
    E_{t}X_{t+1}=D X_t + Z_t 
\end{equation}
where we defined
$$
X_{t}
= \begin{bmatrix}
x_{1,t}\\
x_{2,t}\\
x_{3,t}%
\end{bmatrix}
=V^{-1}%
\begin{bmatrix}
\hat{\pi}_{t} \\ 
\hat{R}_{t-1} \\ 
\hat{b}_{t-1}%
\end{bmatrix}%
=%
\begin{bmatrix}
q_{11}\hat{\pi}_{t}+q_{12}\hat{R}_{t-1} \\ 
q_{21}\hat{\pi}_{t}+q_{22}\hat{R}_{t-1} \\ 
q_{31}\hat{\pi}_{t}+q_{32}\hat{R}_{t-1}+\hat{b}_{t-1}%
\end{bmatrix}
$$
and 
$$
Z_t = 
\begin{bmatrix}
z_{1,t}\\
z_{2,t}\\
z_{3,t}%
\end{bmatrix}
= V^{-1}B \varepsilon_t 
= \begin{bmatrix}
\beta q_{11}+q_{12} & 0 \\ 
\beta q_{21}+q_{22} & 0 \\ 
\beta q_{31}+q_{32}-\varphi _{3} & -1%
\end{bmatrix}%
\begin{bmatrix}
\Delta \theta_{t} \\ 
\psi _{t}%
\end{bmatrix} 
=
\begin{bmatrix}
0 \\
e_2q_{22}\Delta \theta _{t} \\ 
(\beta q_{31} + q_{32} -\varphi_3)\Delta \theta_t -\psi_{t}%
\end{bmatrix}.
$$
Note that in the last equality we used the fact that $\beta q_{21}+q_{22} = e_2q_{22}$.

To find the rational expectation solution, we solve forward the first difference equation associated to the explosive eigenvalue. We distinguish two cases, corresponding to the AM/PF and the PM/AF parametrizations.

\subsubsection*{AM/PF case}
In the AM/PF regime, the eigenvalue $e_{2\text{ }}$ is outside the unit circle. We thus use the second row of \eqref{eq:matsystem_Leeper_diag} to draw linear restrictions between model variables. Let's rewrite the equation as 
$$
x_{2,t} = \frac{1}{e_2}E_tx_{2,t+1}-\frac{1}{e_2}z_{2,t}.
$$
Substituting recursively the future values of $x_{2}$, we obtain 
\begin{equation}\label{eq:x2_fw}
x_{2,t} = -\frac{1}{e_2} \sum_{k=0}^\infty \left(\frac{1}{e_2}\right)^kE_tz_{2,t+k}
        = -q_{22} \sum_{k=0}^\infty \left(\frac{1}{e_2}\right)^kE_t\Delta \theta_{t+k}.
\end{equation}
The monetary policy shock $\theta_t$ follows a stationary AR(1) process
$$\theta_t = \rho_\theta \theta_{t-1}+\nu_{\theta,t},$$ so we have that 
$$
    E_t \Delta \theta_{t+k} = (\rho_\theta-1)\rho_\theta^{k-1}\Delta \theta_t,
$$
which we can plug into \eqref{eq:x2_fw} to get
$$
x_{2,t} = -q_{22}\left(\Delta \theta_t + \sum_{k=1}^\infty\left(\frac{1}{e_2}\right)^k (\rho_\theta-1)\rho_\theta^{k-1}\Delta \theta_{t+k} \right).
$$
After some manipulations, we can obtain 
$$
x_{2,t} = -q_{22} \left(\frac{e_2-1}{e_2-\rho_\theta}\theta_t -\theta_{t-1}\right),
$$
and using the definition $x_{2,t}=q_{21} \hat\pi_t + q_{22}\hat R_{t-1}$ we arrive to 
\begin{align*}
   \hat\pi_t = -\frac{q_{22}}{q_{21}} \hat R_{t-1} -\frac{q_{22}}{q_{21}} \left(\frac{e_2-1}{e_2-\rho_\theta}\theta_t -\theta_{t-1}\right).
\end{align*}
Finally, after substituting $e_2=1+\frac{\beta\phi_{p}}{\pi}$ and noting that $\frac{q_{22}}{q_{21}} = \frac{\pi}{\phi_p}$, we obtain the solution for inflation in the AM/PF regime:
$$ \hat{\pi}_{t} =-\frac{\pi}{\phi_{p}}\hat
{R}_{t-1}-\frac{\beta}{1+\frac{\beta\phi_{p}}{\pi}-\rho_{\theta}}\theta_{t}+\frac{\pi}{\phi_{p}}\theta_{t-1}.$$

\subsubsection*{PM/AF case}
In the PM/AF regime, the eigenvalue $e_{3}$ is outside the unit circle and we need to solve forward the third equation of system \eqref{eq:matsystem_Leeper_diag}. In this case we obtain 
\begin{equation}\label{eq:x3_fw}
 x_{3,t} = -\frac{1}{e_3} \sum_{k=0}^\infty \left(\frac{1}{e_3}\right)^kE_tz_{3,t+k}
        = \frac{1}{e_3} \sum_{k=0}^\infty \left(\frac{1}{e_3}\right)^kE_t\psi_{t+k}- \frac{1}{e_3} \sum_{k=0}^\infty \left(\frac{1}{e_3}\right)^k(\beta q_{31} + q_{32} -\varphi_3)E_t\Delta \theta_{t+k}.
\end{equation}
The fiscal policy shock $\psi_t$ also follows an AR(1) process, so we have that $$E_t\psi_{t+k} = \rho_\psi^{t+k}\psi_t.$$ Substituting this results in \eqref{eq:x3_fw} and following the same steps of the the AM/PF case, we arrive to the relation
$$
x_{3,t}=\frac{1}{e_{3}-\rho _{\psi }}\psi _{t}-\frac{\beta
q_{31}+q_{32}-\varphi _{3} }{e_{3}}\left( \frac{e_{3}-1}{e_{3}-\rho
_{\theta }}\theta _{t}-\theta _{t-1}\right).
$$
Then, we use the definition $x_{3,t} = q_{31}\hat{\pi}_{t}+q_{32}\hat{R}_{t-1}+\hat{b}_{t-1}$ to rewrite the last equation as
$$
    \hat{\pi}_{t} =-\frac{q_{32}}{q_{31}}\hat{R}_{t-1}-\frac{1}{q_{31}}\hat{b}%
_{t-1}+\frac{1}{q_{31}\left( e_{3}-\rho _{\psi }\right) }\psi _{t}-\frac{%
 \beta q_{31}+q_{32}-\varphi _{3} }{q_{31}e_{3}}\left( \frac{%
e_{3}-1}{e_{3}-\rho _{\theta }}\theta _{t}-\theta _{t-1}\right).
$$
We can now define the coefficients
\begin{align*}
    H &= 1-\frac{1}{\beta }+\gamma +\frac{\beta \phi_{p}}{\pi},\\
    K &= \left( \frac{1}{\beta }-\gamma-1\right) \varphi _{1}+\frac{\phi _{p}}{\pi }\varphi_{2},\\
    J &= \beta \varphi _{1}+\left( \frac{1}{\beta }-\gamma-\frac{\beta \phi_{p}}{\pi }\right) \varphi_2,
\end{align*}
so that $q_{31}= \frac{1}{e_3}\frac{K}{H}$ and $q_{32} = \frac{1}{e_3}\frac{J}{H}$. Using these relations, together with $e_3 = \frac{1}{\beta}-\gamma$, the rational expectation solution for inflation in the PM/AF regime can be rewritten as
$$
\hat{\pi}_{t} =-\frac{J}{K}\hat{R}%
_{t-1}-\left(  \frac{1}{\beta}-\gamma\right)  \frac{H}{K}\hat{b}_{t-1}%
+\frac{\left(  \frac{1}{\beta}-\gamma\right)  \frac{H}{K}}{\frac{1}{\beta
}-\gamma-\rho_{\psi}}\psi_{t}+\frac{\left(  \frac{1}{\beta}-\gamma\right)
\frac{H}{K}\varphi_{3}-\frac{J}{K}-\beta}{\frac{1}{\beta}-\gamma}\left(
\frac{\frac{1}{\beta}-\gamma-1}{\frac{1}{\beta}-\gamma-\rho_{\theta}}%
\theta_{t}-\theta_{t-1}\right).
$$

\subsection{A simple DSGE model with PLT} \label{app: sec: DSGE}
\subsubsection{Determinacy analysis with PLT}
The DSGE model in the text is given by equations \eqref{eq:DSGE_Euler}-\eqref{eq:DSGE_govbudgetconstr}
\begin{align*}
\hat{y}_t   &= E_t \hat y_{t+1}-\left(\hat R_t-E_t \hat\pi_{t+1}\right)+ (1−\rho_\varepsilon)\varepsilon_t ,  \\
\hat \pi_t   &= \beta E_t \hat\pi_{t+1}+\kappa \hat y_t  ,   \\
\hat R_{t+1} &= \phi_p\hat \pi_{t+1}\hat R_{t}+\Delta\theta_{t+1}, \\
\hat b_{t+1} &= \frac{1}{\beta}\left(1-\frac{\tau}{b}\gamma\right)\hat b_{t}+\hat R_{t+1}-\frac{1}{\beta}\hat \pi_{t+1}-\frac{1}{\beta}\frac{\tau}{b} \psi_{t+1},
\end{align*}
which, disregarding expectations errors, can be written in matrix form as:
\[
\begin{bmatrix}
-1 & -1 & 0 & 0\\
0 & -\beta & 0 & 0\\
0 & -\phi_{p} & 1 & 0\\
0 & \frac{1}{\beta} & -1 & 1
\end{bmatrix}
\begin{bmatrix}
\hat y_{t+1}\\
\hat \pi_{t+1}\\
\hat R_{t+1}\\
\hat b_{t+1}%
\end{bmatrix}
=
\begin{bmatrix}
-1 & 0 & -1 & 0\\
\kappa & -1 & 0 & 0\\
0 & 0 & 1 & 0\\
0 & 0 & 0 & \frac{1}{\beta}(1-\frac{\tau}{b}\gamma)
\end{bmatrix}
\begin{bmatrix}
\hat y_{t}\\
\hat \pi_{t}\\
\hat R_{t}\\
\hat b_{t}%
\end{bmatrix}
+
\begin{bmatrix}
1-\rho_\varepsilon & 0 & 0\\
0& 0 & 0\\
0& 1 & 0\\
0& 0 & -\frac{1}{\beta}\frac{\tau}{b}%
\end{bmatrix}%
\begin{bmatrix}
\varepsilon_t\\
\theta_{t+1}\\
\psi_{t+1}%
\end{bmatrix}
\]
Inverting the matrix on the left hand side, we obtain%
\begin{align*}
\begin{bmatrix}
\hat y_{t+1}\\
\hat \pi_{t+1}\\
\hat R_{t+1}\\
\hat b_{t+1}%
\end{bmatrix}
&=
\begin{bmatrix}
\frac{\kappa}{\beta}+1 & -\frac{1}{\beta} & 1 & 0\\
-\frac{\kappa}{\beta} & \frac{1}{\beta} & 0 & 0\\
-\frac{\kappa}{\beta}\phi_{p} & \frac{1}{\beta}\phi_{p} & 1 & 0\\
-\frac{\kappa}{\beta^{2}}\left(  \beta\phi_{p}-1\right)   & \frac{1}{\beta
^{2}}\left(  \beta\phi_{p}-1\right)   & 1 & -\frac{1}{\beta}\left(  \frac
{\tau}{b}\gamma-1\right)
\end{bmatrix}
\begin{bmatrix}
\hat y_{t}\\
\hat \pi_{t}\\
\hat R_{t}\\
\hat b_{t}%
\end{bmatrix}
+
\begin{bmatrix}
1-\rho_\varepsilon & 0 & 0\\
0& 0 & 0\\
0& 1 & 0\\
0& 0 & -\frac{1}{\beta}\frac{\tau}{b}%
\end{bmatrix}%
\begin{bmatrix}
\varepsilon_t\\
\theta_{t+1}\\
\psi_{t+1}%
\end{bmatrix}
\\[2ex]
Y_{t+1}  & =AY_{t}+B\epsilon_{t+1}%
\end{align*}
One of the four eigenvalues of the matrix $A\allowbreak$ is
$-\frac{1}{\beta}\left(  \frac{\tau}{b}\gamma-1\right)$.
To study determinacy we have to analyse the other three eigenvalues of the top-left submatrix:
$$\tilde{A}=
\begin{bmatrix}
\frac{\kappa}{\beta}+1 & -\frac{1}{\beta} & 1\\
-\frac{\kappa}{\beta} & \frac{1}{\beta} & 0\\
-\frac{\kappa}{\beta}\phi_{p} & \frac{1}{\beta}\phi_{p} & 1
\end{bmatrix}.
$$

\subsubsection*{AM/PF case}
If $-\frac{1}{\beta}\left(\frac{\tau}{b}\gamma-1\right)  $ is inside
the unit circle, that is, in case of passive fiscal policy ($\gamma_{\tau
}<\frac{b}{\tau}(1+\beta)$), given the presence of two jump variables, to respect Blanchard-Khan conditions for determinacy, $\tilde{A}$ should have one eigenvalue inside and two outside of the unit circle. 
\citet[][appendix C]{Woodford2003} states these conditions in terms of the
determinant, the trace and the sum of the principal minors of $\tilde{A}$, whose values are the following:
\begin{align*}
\det(\tilde{A})  & =\frac{1}{\beta},\\
tr(\tilde{A})  & =2+\frac{1}{\beta}+\frac{\kappa}{\beta},\\
M(\tilde{A})  & =\frac{1}{\beta}\left(  2+\kappa+\beta+\kappa\phi_{p}\right)
.
\end{align*}
In order to have two eigenvalues outside and one inside, one of the following
three cases must hold.
\begin{itemize}

\item Case 1:  two restrictions should be satisfied simultaneously:%
\begin{align*}
1-tr(\tilde{A})+M(\tilde{A})-\det(\tilde{A})  & <0\\
-1-tr(\tilde{A})-M(\tilde{A})-\det(\tilde{A})  & >0
\end{align*}
The first is verified for $\phi_{p}<0$, but the second is never
verified, so this case does not hold.

\item Case 2:  in this case three conditions are required:%
\begin{align*}
1-tr(\tilde{A})+M(\tilde{A})-\det(\tilde{A})  & >0\\
-1-tr(\tilde{A})-M(\tilde{A})-\det(\tilde{A})  & <0\\
\det(\tilde{A})^{2}-\det(\tilde{A})tr(\tilde{A})+M(\tilde{A})-1  & >0
\end{align*}
The first condition is satisfied for $\phi_{p}>0$, the second is always true while the third holds for $\phi_{p}>\frac{1}{\beta}-1$.

\item Case 3:  in this last case four conditions are required:%
\begin{align*}
1-tr(\tilde{A})+M(\tilde{A})-\det(\tilde{A})  & >0\\
-1-tr(\tilde{A})-M(\tilde{A})-\det(\tilde{A})  & <0\\
\det(\tilde{A})^{2}-\det(\tilde{A})tr(\tilde{A})+M(\tilde{A})-1  & <0\\
\left\vert tr(\tilde{A})\right\vert  & >3
\end{align*}
The first is satisfied for $\phi_{p}>0$, the second and the fourth are always
satisfied, the third holds true for $\phi_{p}<\frac{1}{\beta}-1$.
\end{itemize}
In conclusion, from cases 2 and 3 we get the sole condition is $\phi_{p}%
>0$. 
Therefore, when fiscal policy is passive, determinacy can be
achieved for $\phi_{p}>0$.

\subsubsection*{PM/AF case}

If $-\frac{1}{\beta}\left(  \frac{\tau}{b}\gamma-1\right)  $ is
outside the unit circle, that is in case of active fiscal policy
($\gamma>\frac{b}{\tau}(1+\beta)$), given the presence of two jump
variables, to respect Blanchard-Khan conditions for determinacy, $\tilde{A} $
should have two eigenvalues inside and one outside of the unit circle or, equivalently, we would ask for two eigenvalues outside and one inside in the following:
\[
\tilde{A}^{-1}=
\begin{bmatrix}
1 &   1+\phi_{p}  & -1\\
\kappa &   \kappa+\beta+\kappa\phi_{p}  & -\kappa\\
0 & -\phi_{p} & 1
\end{bmatrix}
\]
For this matrix we have
\begin{align*}
\det(\tilde{A}^{-1})  & =\beta\\
tr(\tilde{A}^{-1})  & =\kappa+\beta+\kappa\phi_{p}+2\allowbreak\\
M(\tilde{A}^{-1})  & =2\beta+1+\kappa.
\end{align*}
As before, we should consider the three cases above. We can exclude case 1 since the second condition is never verified. From case 2 we find that the first condition gives $\phi_{p}<0$, the second is always true and the third is verified for $\phi_{p}<\frac{1}{\beta}-1$ (that is a less stringent condition than $\phi_{p}<0$). The third case can be excluded because the third condition gives rise to a contradiction. Therefore, when fiscal policy is active determinacy can be achieved for $\phi_{p}<0$.

\end{document}